\documentclass{article}
\usepackage[margin=1in]{geometry}

\usepackage[utf8]{inputenc}
\usepackage{lmodern}
\usepackage[T1]{fontenc}

\usepackage{amsmath, amssymb, amsthm, amsfonts}

\usepackage[polish, english]{babel}

\usepackage{bbm}
\usepackage{thm-restate}
\usepackage{enumerate}
\usepackage{listings}
\usepackage{verbatim}

\usepackage{tikz}
\usetikzlibrary{calc, arrows, decorations.markings, positioning, fit,
shapes.geometric}
\tikzset{>=stealth'}

\theoremstyle{plain}
\newtheorem{theorem}{Theorem}
\newtheorem{proposition}[theorem]{Proposition}
\newtheorem{lemma}[theorem]{Lemma}
\newtheorem{corollary}[theorem]{Corollary}
\newtheorem{claim}[theorem]{Claim}
\newtheorem{fact}[theorem]{Fact}

\theoremstyle{definition}
\newtheorem{definition}[theorem]{Definition}
\newtheorem{example}[theorem]{Example}
\newtheorem{conjecture}[theorem]{Conjecture}

\theoremstyle{remark}
\newtheorem{remark}[theorem]{Remark}

\DeclareMathOperator*{\EE}{\mathbb{E}}

\DeclareMathOperator{\Ent}{Ent}
\DeclareMathOperator{\wt}{wt}
\DeclareMathOperator{\BEC}{BEC}
\DeclareMathOperator{\BSC}{BSC}

\newcommand{\bbR}{\mathbb{R}}
\newcommand{\bbF}{\mathbb{F}}
\newcommand{\eps}{\varepsilon}

\usepackage{xcolor}
\definecolor{DSgray}{cmyk}{0,0,0,0.7}
\definecolor{DSred}{cmyk}{0,0.7,0,0.7}

\usepackage[pdftex]{hyperref}
\usepackage{cleveref}
\crefname{claim}{claim}{claims}
\crefname{fact}{fact}{facts}
\crefname{conjecture}{conjecture}{conjecture}

\title{Optimal list decoding from noisy entropy inequality}
\author{Jan Hązła\thanks{AIMS Rwanda. Email: {\tt jan.hazla@gmail.com}.\newline 
The author was supported by the AIMS Rwanda research chair funding from
the Alexander von Humboldt Foundation, as well as the DAAD
grant in cooperation with Goethe University in Frankfurt.
}}
\date{}

\begin{document}\maketitle

\begin{abstract}
    A noisy entropy inequality for boolean functions by Samorodnitsky
    is applied to binary codes. It is shown that a binary code
    that achieves capacity on the binary erasure channel
    admits optimal list size for 
    list decoding on some binary symmetric channels
    (in a regime where this optimal list size is exponentially large).
\end{abstract}

In this note we restate the inequality by Samorodnitsky~\cite{Sam16} in terms
of transmitting random binary vectors over BSC and BEC channels
(see also \Cref{sec:interpretation}): 

\begin{theorem}\label{thm:bsc-bec}
Let $X\in\bbF_2^n$ be a random variable. Let $\eta,\eps$ be such that 
$0\le\eta\le 1$ and $4\eps(1-\eps)\ge\eta$.
Denote by $Y_{\BSC}$ and $Y_{\BEC}$ the outputs of transmitting $X$ over BSC$(\eps)$
and BEC$(\eta)$, respectively. Then,
\begin{align*}
    H(X|Y_{\BSC})\le \big(h(\eps)-\eta\big)n+H(X|Y_{\BEC})\;.
\end{align*}
\end{theorem}

We then point out an application to list decoding of binary codes
over the BSC. Consider a family of binary codes 
with increasing blocklength $n$ and rate $R$
that achieve capacity on the BEC 
(in the sense that $H(X|Y_{\BEC})=o(n)$, which implies more
standard conditions
like vanishing bit-MAP or block-MAP error probability\footnote{
In fact, in case of the BEC, achieving capacity in the sense of entropy
is equivalent to achieving capacity for bit-MAP error 
probability \cite{Sam22,Sam22b}.
}).
Then, for $\eps$ such that $4\eps(1-\eps)> 1-R$,
these codes can be list-decoded over BSC$(\eps)$ with vanishing
error probability
for a list of size $\exp_2\Big(\big(R-(1-h(\eps))\big)n+o(n)\Big)$.
This list size is optimal up to the $2^{o(n)}$ factor, as shown in~\cite{RS22}.
Accordingly, our result gives optimal list decoding in the regime
where the optimal list size is exponential, as recently studied
by Rao and Sprumont~\cite{RS22}\footnote{In their list decoding results,
\cite{RS22} focus on transitive codes (which in general
do not achieve capacity on the BEC, as evidenced by the repetition
code) and on the specific case of Reed--Muller codes.
}.

Let $\{C_n\}$ be a family of binary codes indexed by 
increasing blocklengths $n$ with rates 
$R_n:=\log|C_n|$
satisfying $\lim_{n\to\infty}R_n = R$.
We then say that $\{C_n\}$ has rate $R$.

\begin{definition}[Achieving capacity in the sense of entropy]
For the purposes of this note, we say
that $\{C_n\}$ achieves capacity on the BEC
if there exists a sequence $\{\eta_n\}$
with $\lim_{n\to\infty}\eta_n=1-R$ such that
$H(X|Y_{\BEC})=o(n)$, where $X\in\bbF_2^n$ is uniform over $C_n$,
and $Y_{\BEC}$ is the output of transmitting $X$ over
the BEC$(\eta_n)$.
\end{definition}

\begin{theorem}\label{thm:list-decoding}
Let $\{C_n\}$ be a family of binary codes of rate $R$
that achieves capacity on the BEC. Let $\eps$ be such that
$4\eps(1-\eps)>1-R$. 
Suppose that $X\in\bbF_2^n$ is uniform over $C_n$
and $Y_{\BSC}$ is the output of transmitting $X$ over
BSC$(\eps)$.

Then, there exists a function $f(n)=o(n)$ such that for
$k=k_n=\exp_2\Big(\big(R-(1-h(\eps))\big)n+f(n)\Big)$ there
exists a list decoder $D:\bbF_2^n\to C_n^k$ such that
\begin{align*}
    \lim_{n\to\infty}\Pr[X\notin D(Y_{\BSC})] = 0\;.
\end{align*}
\end{theorem}

\Cref{thm:list-decoding} applies to general codes,
in contrast to the result for linear codes from~\cite{HSS21}.
The list decoding error probability goes to zero
only in expectation over a random codeword $X$.
By a standard argument, for linear codes that
implies vanishing error probability also for every
codeword individually.

By Theorem~20 in~\cite{KKM17}, the conclusion of \Cref{thm:list-decoding}
holds for all linear codes with doubly transitive symmetry group.
This includes Reed--Muller codes,
however for those a better (applicable for all $R>1-h(\eps)$) 
bound on $H(X|Y_{\BSC})$ and
optimal list decoding on the BSC follow from the
recent bit-MAP decoding result by Reeves and Pfister~\cite{RP21}.

\section{Preliminaries}

All our logarithms are binary.
Let $f:\bbF_2^n\to\bbR_{\ge 0}$ and $q\ge 1$. 
We define the norms of
$f$ as $\|f\|_q:=\left(\EE_x f(x)^q\right)^{1/q}
=\left(2^{-n}\sum_{x\in\bbF_2^n}f(x)^q\right)^{1/q}$,
as well as the \emph{entropy} of $f$ as
\begin{align*}
    \Ent[f]:=\EE_x f(x)\log f(x) - \left(\EE_x f(x)\right)\cdot \log \EE_x f(x)\;.
\end{align*}
Furthermore, let $X$ be a random variable over $\bbF_2^n$.
We let the distribution function of $X$ to be
$f_X(x):=2^n\Pr[X=x]$. Observe that $f_X$ is nonnegative and that $\|f_X\|_1=1$.

Recall the definition of Rényi entropy:

\begin{definition}[Rényi entropy]
Let $q>1$ and $X$ a discrete random variable.
The \emph{$q$-th Rényi entropy} of $X$ is given by
\begin{align*}
    H_q(X):=-\frac{1}{q-1}\log\sum_{x}\Pr[X=x]^q\;.
\end{align*}
We will also write $H_1(X):=H(X)$ for the Shannon entropy
and $H_\infty(X):=-\log\max_{x}\Pr[X=x]$ for the min-entropy.
\end{definition}

We also let 
$h_q(\eps):=-\frac{1}{q-1}\log(\eps^q+(1-\eps)^q)$ for $0\le\eps\le 1$.
This is the $q$-th Rényi entropy of the one-bit 
Bernoulli($\eps$) distribution. Similarly, we let
$h(\eps):=h_1(\eps)=-\eps\log\eps-(1-\eps)\log(1-\eps)$
and $h_{\infty}(\eps):=-\log\max(\eps,1-\eps)$.

The following facts hold by simple calculations:
\begin{fact}\label{fac:entropies}
For every $1\le q\le\infty$:
\begin{enumerate}
    \item $H_q(X)=0$ if and only if $X$ is deterministic.
    \item If $X\in \mathcal{X}$, then $H_q(X)=\log |\mathcal{X}|$
    if and only if $X$ is uniform over $\mathcal{X}$.
\end{enumerate}
\end{fact}

It is a known fact that Rényi entropies decrease with $q$:
\begin{proposition}
If $X$ is a random variable and $1\le p\le q\le\infty$, then
$H_p(X)\ge H_q(X)$.
\end{proposition}

\begin{claim}
Let $f_X$ be a distribution function of a random variable $X$ over $\bbF_2^n$. 
Then,
\begin{align}\label{eq:07}
\Ent[f_X]=n-H(X)\;.
\end{align}
Similarly, for $1<q\le\infty$,
\begin{align}\label{eq:08}
\log\|f_X\|_q=\frac{q-1}{q} (n-H_q(X))\;.
\end{align}
\end{claim}

\begin{proof}
Let us start with $q>1$. Indeed, we have
\begin{align*}
    \log\|f_X\|_q
    =\frac{1}{q}\log\left[2^{-n}\sum_x (2^n\Pr[X=x])^q\right]
    =\frac{q-1}{q}n+\frac{1}{q}\log\sum_x\Pr[X=x]^q
    =\frac{q-1}{q}(n-H_q(X))\;.
\end{align*}
For the Shannon entropy, we see in a similar way,
\begin{align*}
    \Ent[f_X]&=
    \EE_x f_X(x)\log f_X(x) -\left(\EE_x f_X(x)\right)\cdot \log \EE_x f_X(x)
    =\EE_x f_X(x)\log f_X(x)\\
    &=2^{-n}\sum_x 2^n\Pr[X=x]\log(2^n\Pr[X=x])
    =n-H(X)\;.
\end{align*}
\end{proof}

\section{Inequalities by Samorodnitsky}

We introduce more notation to discuss the inequalities
by Samorodnitsky.
First, for $0\le\eps\le 1$, we define
the operator acting on functions $f:\bbF_2^n\to\mathbb{R}$:
\begin{align*}
    T_\eps f(x):=
    \sum_{y\in\bbF_2^n}
    \eps^{|y|}
    (1-\eps)^{n-|y|}
    f(x+y)\;,
\end{align*}
where the addition $x+y$ is in the vector
space $\mathbb{F}_2^n$, and $|y|$ denotes the number of nonzero
coordinates in $y$. 

Furthermore, for $S\subseteq[n]$ we define 
$\EE(f|S):\bbF_2^S\to\bbR$ as the conditional expectation
of $f$ with respect to $S$, ie.,
$\EE(f\vert S)(x_S)=\EE_{y:y_S=x_S} f(y)$,
where $y_S\in\bbF_2^S$ denotes vector $y$ 
restricted to the coordinates in $S$
and the distribution over $y\in\bbF_2^n$ is uniform. 

The inequality of Samorodnitsky for $q$-norms reads:
\begin{theorem}[Theorem~1.1 in~\cite{Sam20}]
\label{thm:sam}
Let $q\ge 2$ be integer and $f:\bbF_2^n\to\bbR_{\ge 0}$. Then,
\begin{align}\label{eq:02}
    \log\|T_\eps f\|_q\le \EE_{S\sim\lambda}\log\|\EE(f|S)\|_q\;,
\end{align}
where $\lambda:=1-h_q(\eps)$,
and $S\sim \lambda$ denotes a random subset in $[n]$ where each
element is put in $S$ independently with probability $\lambda$.
\end{theorem}

This inequality comes from a line of work~\cite{Sam16, Sam19} where,
among others, an
analogous inequality was proved for entropies:
\begin{theorem}[Corollary 1.9 in~\cite{Sam16}]
\label{thm:sam-entropy}
Let $f:\bbF_2^n\to\bbR_{\ge 0}$. Then,
\begin{align}\label{eq:09}
    \Ent[T_\eps f]\le \EE_{S\sim\lambda}\Ent[\EE(f|S)]\;,
\end{align}
where $\lambda:=(1-2\eps)^2$.
\end{theorem}

\section{Entropic interpretation of the inequalities}
\label{sec:interpretation}

In this section we restate \Cref{thm:sam,thm:sam-entropy} in terms
of respective entropies.
While similar statements and some of their consequences have
been given before, see~\cite{Sam16,Ord16,PW17,Sam19,HSS21,Sam22},
we find it instructive to put them together for
comparison. In the proofs of \Cref{thm:bsc-bec,thm:list-decoding}
we will only need \Cref{cor:rv-entropy}.

\begin{claim}\label{cl:e-fx}
Let $X\in\bbF_2^n$ be a random variable
and $S\subseteq[n]$. 
Then, $\EE(f_X|S)=f_{X_S}$, where $X_S\in\bbF_2^S$
is $X$ restricted to the coordinates in $S$.
\end{claim}
\begin{proof}
$\EE(f_X|S)(x_S)=2^{|S|-n}\sum_{y\in\bbF_2^n:y_S=x_S}2^n\Pr[X=y]
=2^{|S|}\Pr[X_S=x_S]=f_{X_S}(x_S)$.
\end{proof}

\begin{claim}\label{cl:t-fx}
Let $X\in\mathbb{F}_2^n$ be a random variable
and $Z\in\mathbb{F}_2^n$ be iid Ber$(\eps)$ and independent
of $X$. Then, $T_\eps f_X=f_{X+Z}$.
\end{claim}
\begin{proof}
$T_\eps f_X(y)=\sum_z\eps^{|z|}(1-\eps)^{n-|z|}f_X(y+z)
=2^n\sum_z\Pr[Z=z,X=y+z]=2^n\Pr[X+Z=y]=f_{X+Z}(y)
$.
\end{proof}

We can now state two corollaries (in fact, equivalent
formulations) of Samorodnitsky's inequalities.

\begin{corollary}\label{cor:rv}
Let $X\in\bbF_2^n$ be a random variable and let
$Z\in\bbF_2^n$ be iid Ber$(\eps)$ independent of $X$. 
Then, for every integer $q\ge 2$ it holds
\begin{align*}
    H_q(X+Z)\ge(1-\lambda)n+\EE_{S\sim\lambda}H_q(X_S)\;,
\end{align*}
where $\lambda=\lambda(q)=1-h_q(\eps)$.
\end{corollary}

\begin{proof}
Let $f_X$ be the distribution function of $X$.
By \Cref{thm:sam},
we have $\log\|T_\eps f_X\|_q\le \EE_{S\sim\lambda}\log\|\EE(f_X|S)\|_q$.
But by \Cref{cl:e-fx,cl:t-fx},
$T_\eps f_X=f_{X+Z}$ and
$\EE(f_X|S)=f_{X_S}$. Therefore,
$\log\|f_{X+Z}\|_q\le\EE_{S\sim\lambda}\log\|f_{X_S}\|_q$.
Applying~\eqref{eq:08} to both sides 
and rearranging concludes the proof.
\end{proof}

\begin{corollary}\label{cor:rv-entropy}
Let $X\in\bbF_2^n$ be a random variable and let
$Z\in\bbF_2^n$ be iid Ber$(\eps)$ independent of $X$. 
Then,
\begin{align*}
    H(X+Z)\ge(1-\lambda)n+\EE_{S\sim\lambda}H(X_S)\;,
\end{align*}
where $\lambda=(1-2\eps)^2$.
\end{corollary}

\begin{proof}
The same as for \Cref{cor:rv}, using
\Cref{thm:sam-entropy} and~\eqref{eq:07}.
\end{proof}

\section{Proof of \Cref{thm:bsc-bec}}

\begin{claim}\label{cl:bec-conditional-entropy}
Let $X\in\bbF_2^n$ be a random variable.
Then,
\begin{align*}
    \EE_{S\sim\lambda}H(X_S)=H(X)-H(X|Y_{\BEC})\;,
\end{align*}
where $Y_{\BEC}$ is the output of transmitting $X$ over
BEC$(1-\lambda)$.
\end{claim}

\begin{proof}
Let $S\subseteq[n]$. By the chain rule,
$H(X_S)=H(X)-H(X|X_S)$.
But $H(X|X_S)$ is the conditional entropy of $X$
given that its coordinates outside of $S$ have been
erased. Therefore, $\EE_{S\sim\lambda} H(X|X_S)$
is the conditional entropy of $X$ given the output of its transmission
over BEC$(1-\lambda)$. The conclusion follows.
\end{proof}

\begin{proof}[Proof of \Cref{thm:bsc-bec}]
Let $Y_{\BSC}=X+Z$, where $Z$ is iid Ber$(\eps)$ and independent of $X$.
By the chain rule, 
$H(X|Y_{\BSC})=H(X,Y_{\BSC})-H(Y_{\BSC})=H(X,Z)-H(Y_{\BSC})=
H(X)+h(\eps)n-H(X+Z)$.

Let $\eps_0$ be such that $\eta=4\eps_0(1-\eps_0)$ and
let $\lambda:=1-\eta$. Note that $\lambda=(1-2\eps_0)^2$.
Let $Z'$ be 
iid Ber$(\eps_0)$ and independent of $X$. Since
$4\eps(1-\eps)\ge\eta=4\eps_0(1-\eps_0)$, the distribution of $Z'$ 
can be degraded to the distribution of $Z$ and
$H(X+Z)\ge H(X+Z')$, and, by \Cref{cor:rv-entropy},
$H(X+Z')\ge(1-\lambda)n+\EE_{S\sim\lambda}H(X_S)$. 

Chaining it together and applying \Cref{cl:bec-conditional-entropy},
\begin{align*}
    H(X|Y_{\BSC})&=H(X)+h(\eps)n-H(X+Z)
    \le H(X)+h(\eps)n-\eta n-\EE_{S\sim 1-\eta}H(X_S)\\
    &=\big(h(\eps)-\eta\big)n+H(X|Y_{\BEC})\;.\qedhere
\end{align*}
\end{proof}

\section{Proof of \Cref{thm:list-decoding}}

Let $\{C_n\}$ be a family of binary codes of rate $R$ that
achieves capacity on the BEC. 
Accordingly, let $\{\eta_n\}$ be such that
$\eta_n\to 1-R$ and $H(X|Y_{\BEC})=o(n)$,
where $X$ is uniform over $C_n$ and $Y_{\BEC}$ is the output
of transmitting $X$ over BEC$(\eta_n)$.
As before, let $Y_{\BSC}=X+Z$ be the output of transmitting
$X$ over BSC$(\eps)$, so that $Z$ is iid
Ber$(\eps)$ and independent of $X$.

\begin{claim}\label{cl:entropy-capacity}
$H(X|Y_{\BSC})\le
\big(h(\eps)-(1-R)\big)n+o(n)$ and consequently
$H(Y_{\BSC})=n-o(n)$.
\end{claim}

\begin{proof}
Since $\eta_n\to 1-R$, for large enough $n$ we have
$4\eps(1-\eps)\ge\eta_n$. Applying \Cref{thm:bsc-bec},
\begin{align*}
    H(X|Y_{\BSC})&\le\big(h(\eps)-\eta_n\big)n+H(X|Y_{\BEC})
    =\big(h(\eps)-(1-R)\big)n+o(n)\;.
\end{align*}
By chain rule, also
$H(Y_{\BSC})=H(X,Y_{\BSC})-H(X|Y_{\BSC})=H(X,Z)-H(X|Y_{\BSC})
=Rn+h(\eps)n-H(X|Y_{\BSC})\ge n-o(n)$.
\end{proof}

\begin{remark}
Conversely, $H(X|Y_{\BSC})=H(X,Z)-H(Y_{\BSC})
\ge Rn+h(\eps)n-n$, so the bounds in 
\Cref{cl:entropy-capacity} are tight up
to the $o(n)$ terms.
\end{remark}

\begin{claim}\label{cl:partial-entropy-upper-bound}
Let $0\le p_1,\ldots, p_k$ be such that $\sum_{i=1}^k p_i=:p\le 1$.
Then $\sum_{i=1}^k p_i\log\frac{1}{p_i}\le p\log k+1$.
\end{claim}

\begin{proof}
$\sum_{i=1}^k p_i\log\frac{1}{p_i}=p\sum_{i=1}^k\frac{p_i}{p}\left(\log\frac{p}{p_i}+\log\frac{1}{p}\right)
=p\left(\sum_{i=1}^k\frac{p_i}{p}\log\frac{p}{p_i}\right)+p\log\frac{1}{p}\le p\log k+1$,
where in the end we use the fact that the remaining sum is the Shannon entropy of 
a probability distribution over 
$k$ elements.
\end{proof}

By \Cref{cl:entropy-capacity}, \Cref{thm:list-decoding} follows 
from the following lemma:

\begin{lemma}
Let $\{C_n\}$ be a family of binary codes of rate $R$
such that $H(Y_{\BSC})=n-o(n)$ for some channel
BSC$(\eps)$.

Then, there exists a function $f(n)=o(n)$ such that for
$k=k_n=\exp_2\Big(\big(R-(1-h(\eps))\big)n+f(n)\Big)$ there
exists a list decoder $D:\bbF_2^n\to C_n^k$ such that
\begin{align*}
    \lim_{n\to\infty}\Pr[X\notin D(Y_{\BSC})] = 0\;.
\end{align*}
\end{lemma}

\begin{proof}
For the rest of the proof let us write $Y:=Y_{\BSC}=X+Z$.
We shall assume that $\eps\le 1/2$, since the case $\eps>1/2$
follows by relabeling ones and zeros.
Let $\mathcal{L}:=\{(x,z): x\in C_n, z\in\bbF_2^n, \wt(z)<\eps n+n^{3/4}\}$,
that is $\mathcal{L}$ contains the (somewhat) more likely values for $(X,Z)$. For $y\in\bbF_2^n$, let
$\mathcal{B}_y:=\{(x,z):x\in C_n,z\in\bbF_2^n, x+z=y\}$.

Fix $\delta>0$. 
Let us say that $y\in\bbF_2^n$ is $\delta$-likely if the size
of $\mathcal{B}_y\cap\mathcal{L}$ is more
than $\exp_2\Big(\big(R_n-(1-h(\eps))+\delta\big)n\Big)$.
Let $p_n(\delta)$ denote the probability that the
random string $Y\in\mathbb{F}_2^n$ is $\delta$-likely.

If $(x,z)\in\mathcal{L}$, it follows that
\begin{align*}
    \Pr[X=x,Z=z]>
    2^{-R_n n}\cdot \eps^{\eps n+n^{3/4}}\cdot (1-\eps)^{(1-\eps)n-n^{3/4}}
    =\exp_2\Big(-\big(R_n+h(\eps)\big) n)\Big)\cdot\left(\frac{\eps}{1-\eps}\right)^{n^{3/4}}\;.
\end{align*}
Hence, if $y$ is $\delta$-likely, then
$\Pr[Y=y]$ is more than
$\exp_2\big(-(1-\delta)n\big)\cdot\left(\frac{\eps}{1-\eps}\right)^{n^{3/4}}$.
This implies a bound on the entropy of $Y$
\begin{align*}
    H(Y)\le p_n(\delta)\left((1-\delta)n+n^{3/4}\log\frac{1-\eps}{\eps}\right)
    +\sum_{y:\text{$y$ is not $\delta$-likely}}\Pr[Y=y]\cdot \log\frac{1}{\Pr[Y=y]}\;.
\end{align*}
By \Cref{cl:partial-entropy-upper-bound}, the contribution to $H(Y)$ from strings $y$ which are
not $\delta$-likely is at most $(1-p_n(\delta))n+1$ (since there are trivially
at most $2^n$ such strings), and hence
\begin{align}\label{eq:12}
    H(Y)\le n-p_n(\delta)\cdot\delta n+1+n^{3/4}\log\frac{1-\eps}{\eps}\;.
\end{align}

If $\limsup_{n\to\infty} p_n(\delta)>0$,
then~\eqref{eq:12} implies that $n-H(Y)=\Omega(n)$ infinitely often,
which is a contradiction with the assumption $H(Y)=n-o(n)$. Therefore, we have $\lim_{n\to\infty}p_n(\delta)=0$
for every fixed $\delta>0$. It follows that there is a
sequence $\{\delta_n\}$ such that
$\lim_{n\to\infty}\delta_n=\lim_{n\to\infty}p_n(\delta_n)=0$.

Let $k$ be equal to the (ceiling of)
$\exp_2\Big(\big(R_n-(1-h(\eps))+\delta_n\big)n\Big)$. By the
assumption $R_n\to R$, indeed we have
$k=\exp_2\Big(\big(R-(1-h(\eps))\big)n+o(n)\Big)$.
Let $D$ be any decoder that, for each $y\in\bbF_2^n$ which is \emph{not} $\delta_n$-likely,
outputs a list that contains all codewords
$x\in C_n$ such that $(x,x+y)\in\mathcal{B}_y\cap\mathcal{L}$. Note that we defined $k$
to be large enough for this to be possible.

By construction, if $Y$ is not $\delta_n$-likely and
$(X,Z)\in\mathcal{L}$, then the decoder is successful and
$X\in D(Y)$. Therefore, the error probability is bounded by
\begin{align*}
    \Pr[X\notin D(Y)]
    \le\Pr[\text{$Y$ is $\delta_n$-likely}]+\Pr[(X,Z)\in\mathcal{L}]=
    p_n(\delta_n)+\Pr[\wt(Z)<\eps n+n^{3/4}]=o(1)\;,
\end{align*}
since $p_n(\delta_n)$ vanishes by construction
and $\Pr[\wt(Z)<\eps n+n^{3/4}]=\exp(-\Omega(\sqrt{n}))$
by the Hoeffding's inequality.
\end{proof}

\begin{remark}
As mentioned, the list decoding size
$k=\exp_2\Big(\big(R-(1-h(\eps))\big)n\Big)$
is optimal for a code of rate $R$ transmitted over BSC$(\eps)$, up to $2^{o(n)}$ factor.
This follows immediately from the claim in~\cite{RS22} that we restate
below.
\end{remark}

\begin{claim}[Claim 30 in \cite{RS22}]
Let $0<\eps<1/2$ and $n>10/\eps^2$. Let
$C\subseteq \bbF_2^n$ be a code of rate $R$, 
$X$ be uniform over $C$
and $Y_{\BSC}$ the result of transmitting $X$
over BSC$(\eps)$.

Suppose a decoder $D:\bbF_2^n\to C^k$ satisfies $\Pr[X\in D(Y_{\BSC})]\ge 3/4$. Then,
we must have
\begin{align*}
    k\ge \exp_2\Big(\big(R-(1-h(\eps))\big)n-h(\eps)n^{3/4}-3\Big)\;.
\end{align*}
\end{claim}

\paragraph{Acknowledgements}
I am grateful to Emmanuel Abbe and Alex Samorodnitsky for helpful conversations.

\bibliographystyle{alpha}
\bibliography{biblio_arxiv}

\newcommand{\etalchar}[1]{$^{#1}$}
\begin{thebibliography}{KKM{\etalchar{+}}17}

\bibitem[HSS21]{HSS21}
Jan H{\k{a}}z{\l}a, Alex Samorodnitsky, and Ori Sberlo.
\newblock On codes decoding a constant fraction of errors on the {BSC}.
\newblock In {\em Symposium on Theory of Computing (STOC)}, pages 1479--1488,
  2021.

\bibitem[KKM{\etalchar{+}}17]{KKM17}
Shrinivas Kudekar, Santhosh Kumar, Marco Mondelli, Henry~D Pfister, Eren
  {\c{S}}a{\c{s}}oǧlu, and Rüdiger Urbanke.
\newblock Reed--{M}uller codes achieve capacity on erasure channels.
\newblock {\em IEEE Transactions on Information Theory}, 63(7):4298--4316,
  2017.

\bibitem[Ord16]{Ord16}
Or~Ordentlich.
\newblock Novel lower bounds on the entropy rate of binary hidden {M}arkov
  processes.
\newblock In {\em International Symposium on Information Theory (ISIT)}, pages
  690--694, 2016.

\bibitem[PW17]{PW17}
Yury Polyanskiy and Yihong Wu.
\newblock Strong data-processing inequalities for channels and {B}ayesian
  networks.
\newblock In {\em Convexity and Concentration}, pages 211--249. Springer, 2017.

\bibitem[RP21]{RP21}
Galen Reeves and Henry~D Pfister.
\newblock {R}eed--{M}uller codes achieve capacity on {BMS} channels.
\newblock arXiv:2110.14631v2, 2021.

\bibitem[RS22]{RS22}
Anup Rao and Oscar Sprumont.
\newblock A criterion for decoding on the {BSC}.
\newblock arXiv:2202.00240v5, 2022.

\bibitem[Sam]{Sam22b}
Alex Samorodnitsky.
\newblock Personal communication.

\bibitem[Sam16]{Sam16}
Alex Samorodnitsky.
\newblock On the entropy of a noisy function.
\newblock {\em IEEE Transactions on Information Theory}, 62(10):5446--5464,
  2016.

\bibitem[Sam19]{Sam19}
Alex Samorodnitsky.
\newblock An upper bound on $\ell_q $ norms of noisy functions.
\newblock {\em IEEE Transactions on Information Theory}, 66(2):742--748, 2019.

\bibitem[Sam20]{Sam20}
Alex Samorodnitsky.
\newblock An improved bound on $\ell_q$ norms of noisy functions.
\newblock arXiv:2010.02721v1, 2020.

\bibitem[Sam22]{Sam22}
Alex Samorodnitsky.
\newblock On some properties of random and pseudorandom codes.
\newblock arXiv:2206.05135v1, 2022.

\end{thebibliography}

\end{document}